\documentclass[a4paper, 11pt]{article}
\usepackage{amsthm,amsmath,amssymb}
\usepackage{subcaption}
\usepackage[usenames,dvipsnames]{color}
\usepackage[pdftex,breaklinks,colorlinks,
    citecolor={BlueViolet}, linkcolor={Blue},urlcolor=Maroon]{hyperref}
\usepackage{tikz,pgfkeys}
\usepackage{tkz-graph}
\usetikzlibrary{arrows.meta}
\usetikzlibrary{quotes}
\usetikzlibrary{decorations.pathreplacing}
\usepackage{graphicx}
\usepackage{charter,eulervm}%
\usepackage{multirow,booktabs,array}
\usepackage{enumerate}
\usepackage[final,expansion=alltext,protrusion=true]{microtype}
\usepackage{tikz}

\theoremstyle{plain} 
\newtheorem{theorem}{Theorem}[section]
\newtheorem{lemma}[theorem]{Lemma}

\newtheorem{proposition}[theorem]{Proposition}
\newtheorem{definition}[theorem]{Definition}

\tikzstyle{vertex}  = [{circle,blue,draw,fill=black!50,inner sep=1pt}]  
\tikzstyle{uvertex} = [{violet,draw,fill=violet!50,inner sep=2pt}]  

\newcommand{\lp}[1]{\ensuremath{{\mathtt{lp}(#1)}}}
\newcommand{\rp}[1]{\ensuremath{{\mathtt{rp}(#1)}}}

\title{Characterization and Linear-time Recognition \\ of Paired Threshold Graphs}
\author{Guozhen Rong\thanks{School of Computer Science and Engineering, Central South University, Changsha, China.}
  \and
  Yixin Cao\thanks{Department of Computing, Hong Kong Polytechnic University, Hong Kong, China. \href{mailto:yixin.cao@polyu.edu.hk} {\tt yixin.cao@polyu.edu.hk}.  {Supported in part by the Hong Kong Research Grants Council (RGC) under grants 15201317 and 15226116, and the National Natural Science Foundation of China (NSFC) under grant 61972330.}} 
  \and
  Jianxin Wang\footnotemark[1]
}

\begin{document}
\maketitle

\begin{abstract}
  In a paired threshold graph, each vertex has a weight, and two vertices are adjacent if their weight sum is large enough and their weight difference is small enough.  It generalizes threshold graphs and unit interval graphs, both very well studied.
  We present a vertex ordering characterization of this graph class, which enables us to prove  that it is a subclass of interval graphs.  Further study of clique paths of paired threshold graphs leads to a simple linear-time recognition algorithm for the class. 

\end{abstract}

\section{Introduction}

A graph is a \emph{threshold graph} if one can assign positive weights to its vertices in a way that two vertices are adjacent if and only if the sum of their weights is not less than a certain threshold.  Originally formulated from combinatorial optimization \cite{chvatal-77-inequalities-in-ip}, threshold graphs found applications in many diversified areas.  As one of the simplest nontrivial classes, the mathematical properties of threshold graphs have been thoroughly studied.  They admit several nice characterizations, including inductive construction, degree sequences, forbidden induced subgraphs (Figure~\ref{fig:threshold-free}), to name a few \cite{mahadev-95-threshold-graphs}.
Relaxing these characterizations in one way or another, we end with several graph classes, e.g., cographs, split graphs, trivially perfect graphs and double-threshold graphs \cite{nikolopoulos-00-cographs-threshold-graphs,foldes-77-split-graphs,golumbic-2004-perfect-graphs, kobayashi-19-double-threshold-graphs}.  Yet another closely related graph class is the difference graphs, defined solely by weight differences \cite{hammer-90-difference-graphs}.

\begin{figure}[h]
  \centering      
  \footnotesize
  \begin{subfigure}[b]{0.15\linewidth}
    \centering
    \begin{tikzpicture}[every node/.style={vertex},scale=.5]
      \node (a) at (-1,0) {};
      \node (c) at (1,0) {};
      \node (b) at (-1,2) {};
      \node (d) at (1,2) {};
      \draw (a) -- (b) (c) -- (d);
    \end{tikzpicture}
    \caption{$2K_2$}
  \end{subfigure}
  \begin{subfigure}[b]{0.15\linewidth}
    \centering
    \begin{tikzpicture}[every node/.style={vertex},scale=.5]
      \node (a) at (-1,0) {};
      \node (c) at (1,0) {};
      \node (b) at (-1,2) {};
      \node (d) at (1,2) {};
      \draw (b) -- (a) -- (c) -- (d);
    \end{tikzpicture}
    \caption{$P_4$}
  \end{subfigure}
  \begin{subfigure}[b]{0.15\linewidth}
    \centering
    \begin{tikzpicture}[every node/.style={vertex}, scale=.5]
      \node (a) at (-1,0) {};
      \node (c) at (1,0) {};
      \node (b) at (-1,2) {};
      \node (d) at (1,2) {};
      \draw (a) -- (b) -- (d) -- (c) -- (a);
    \end{tikzpicture}
    \caption{$C_4$}
  \end{subfigure}
  \caption{Minimal forbidden induced subgraphs for threshold graphs.  }
  \label{fig:threshold-free}
\end{figure}
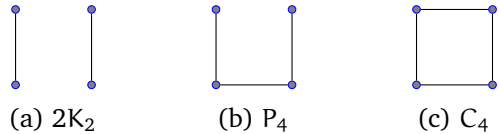

Motivated by applications in social and economic interaction modeling, Ravanmehr et al.~\cite{ravanmehr-18-paired-threshold-graphs} introduced paired threshold graphs, another generalization of threshold graphs.
A graph is a \emph{paired threshold graph} if there exist a positive vertex weight assignment $w$ and two positive thresholds, $T_+$ and $T_-$, such that two vertices are adjacent if and only if the sum of their weights are not less than $T_+$ and the difference of their weights are not greater than $T_-$.

An easy observation on a threshold graph is, vertices of small weights, less than half of the threshold to be specific, form an independent set, while other vertices form a clique.  (Hence, each threshold graph is a split graph.)  Clearly, the first part remains true for paired threshold graphs, but not the second.
Since the adjacency between a pair of high-weight vertices ($\ge \frac{T_+}{2}$) is only decided by their weight difference, they induce an indifference graph \cite{roberts-69-indifference-graphs}, which is more widely known as a unit interval graph.  The crucial point is thus to understand the interaction between these two sets of vertices.  For this purpose we may focus on paired threshold graphs that are neither threshold graphs nor unit interval graphs.
Ravanmehr et al.~\cite{ravanmehr-18-paired-threshold-graphs} presented a distance decomposition for such a paired threshold graph $G$: If $G$ is connected, then they are able to decompose $V(G)$ into a set $X$, which induces a threshold graph, and a sequence of cliques, of each of which the vertices have the same distance to $X$.

It is straightforward to show that paired threshold graphs are chordal: The vertex with the smallest weight is necessarily simplicial.
Since interval graphs also contain all threshold graphs and all unit interval graphs, a natural question is on the relationship between interval graphs and paired threshold graphs.  

\begin{theorem} \label{thm:PTisInterval}
  All paired threshold graphs are interval graphs.
\end{theorem}

Threshold graphs enjoy a very simple ordering characterization by the vertex degrees \cite{chvatal-77-inequalities-in-ip}, while the ordering of the intervals gives a vertex ordering characterization for unit interval graphs, called an umbrella ordering \cite{looges-93-greedy-algorithms-uig}.  On the other hand, interval graphs have a vertex ordering characterization with the so-called 3-vertex conditions \cite{ramalingam-88-domination-interval-graphs}.  We show that a paired threshold graph admits an interval ordering with the additional conditions: 
(1) it can be partitioned such that the first part induces an independent set and the second is an umbrella ordering; 
and (2) the neighborhood of every vertex from the first part is consecutive in this ordering.
We call such a vertex ordering a broom ordering.
\begin{theorem}\label{thm:broom-ordering-for-pt}
  A graph is a paired threshold graph if and only if it admits a broom ordering.
\end{theorem}

Unit interval graphs are interval graphs that can be represented using intervals of a single length.
It is known that any threshold graph can be represented by intervals of at most two different lengths \cite{golumbic-2004-perfect-graphs}.  (But not all interval graphs with two-length representation are threshold graphs.)
This is nevertheless not true for paired threshold graphs.  For each $k > 0$, we are able to construct paired threshold graphs that cannot be represented by intervals of $k$ different lengths.  In other words, the class of paired threshold graphs is not a subclass of $k$-length interval graphs, defined by Klav{\'{\i}}k et al.~\cite{klavik-16-nested-interval-graphs}.  
Recall that unit interval graphs are also proper interval graphs, interval graphs that can be represented using intervals none of which properly contains the other.  This has also been generalized by Klav{\'{\i}}k et al.~\cite{klavik-16-nested-interval-graphs}, who defined the classes of $k$-nested interval graphs, for $k \in \mathbb{N}$, to be the interval graphs that can be represented using intervals of which no $k$ nested.
Indeed, we show that there must be $k$ nested intervals in any interval model of this graph.
Therefore, the class of paired threshold graphs and the class of $k$-nested interval graphs are not comparable to each other.
See Figure~\ref{fig:overview} for an overview of related graph classes.

\tikzstyle{class} = [shape=rectangle, minimum height=5mm, rounded corners, draw, align=center, top color=white, bottom color=blue!20]
\newcommand{\hsc}[1]{{\footnotesize\sf\MakeUppercase{#1}}}
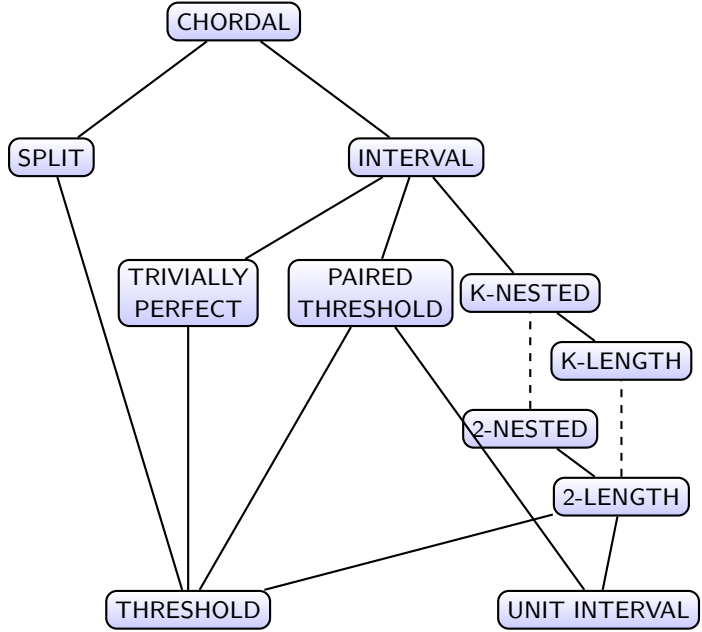
\begin{figure}[h]
  \centering\small
  \begin{tikzpicture}[every path/.style={thick}, scale = 1.2]
    \node[class] (chordal) at (-.5, 7) {\hsc{chordal}};
    \node[class] (split) at (-2.5, 5.5) {\hsc{split}};
    \node[class] (interval) at (1.5,5.5) {\hsc{interval}};
    \node[class] (threshold) at (-1,0.5) {\hsc{threshold}};
    \node[class] (tpg) at (-1,4) {\hsc{trivially}\\\hsc{perfect}};
    \node[class] (uig) at (3.5,0.5) {\hsc{unit interval}};
    \node[class] (paired) at (1,4) {\hsc{paired}\\\hsc{threshold}};
    
    \node[class] (l2) at (3.75,1.75) {\hsc{2-length}};
    \node[class] (n2) at (2.75,2.5) {\hsc{2-nested}};
    \node[class] (l3) at (3.75,3.25) {\hsc{k-length}};
    \node[class] (n3) at (2.75,4.) {\hsc{k-nested}};
    
    \draw (threshold) -- (split);
    \draw (split) -- (chordal);
    \draw (interval) -- (chordal);
    \draw (threshold) -- (paired) (uig) -- (paired) -- (interval); 
    \draw (threshold) -- (tpg) -- (interval);
    \draw (threshold) -- (l2) -- (n2);
    \draw (uig) -- (l2) (l3) -- (n3) -- (interval);
    \draw[dashed] (l2) -- (l3) (n2) -- (n3);
  \end{tikzpicture}

  \caption{A summary of related graph classes, especially subclasses of interval graphs.  For $k \in \mathbb{N}$, the classes of $k$-length interval graphs and $k$-nested interval graphs are defined in \cite{klavik-16-nested-interval-graphs}.  Note that the three immediate subclasses of interval graphs, namely, trivially perfect graphs, paired threshold graphs, and $k$-nested interval graphs are not comparable to each other.}
  \label{fig:overview}
\end{figure}

Similar as threshold graphs and split graphs, the class of paired threshold graphs is not closed under taking disjoint union of subgraphs.
If a paired threshold graph is not a unit interval graph, then in any assignment, there must be some vertex receiving weight $< T_+/2$ and some vertex receiving weight $\ge T_+/2$.  (Note that an independent set is trivially a unit interval graph.)  From the definition it is easy to verify that at most one component can be a non-unit interval graphs, which is of course a connected paired threshold graph. This turns out to be also sufficient.

For the recognition of paired threshold graphs, we may focus on connected non-unit interval graphs.  For such a graph, we show that it is a paired threshold graph if and only if an induced subgraph with certain property is.
From this subgraph, we can produce two partitions of its vertex set, and it is a paired threshold graph if and only if one of them defines a broom ordering of this subgraph.
Putting this together, we develop a linear-time algorithm for recognizing paired threshold graphs, improving the $O(|V(G)|^6)$-time algorithm of Ravanmehr et al.~\cite{ravanmehr-18-paired-threshold-graphs}.

\begin{theorem}\label{thm:RecognizePTG}
  Given a graph $G$, we can decide in $O(|V(G)| + |E(G)|)$ time whether $G$ is a paired threshold graph.
\end{theorem}

\section{Preliminaries}
All graphs discussed in this paper are undirected and simple.  The vertex set and edge set of a graph $G$ are denoted by, respectively, $V(G)$ and $E(G)$.
For a subset $X\subseteq V(G)$, denote by $G[X]$ the subgraph of $G$ induced by $X$, and by $G - X$ the subgraph $G[V(G)\setminus X]$; when $X$ consists of a single vertex $v$, we use $G - v$ as a shorthand for $G - \{v\}$.

The \emph{(open) neighborhood} of a vertex $v \in V(G)$, denoted by $N(v)$, comprises vertices adjacent to $v$, i.e., $N(v) = \{ u \mid uv \in E(G) \}$, and the \emph{closed neighborhood} of $v$ is $N[v] = N(v) \cup \{ v \}$.
The \emph{closed neighborhood} and the \emph{(open) neighborhood} of a set $X\subseteq V(G)$ of vertices are defined as $N[X] = \bigcup_{v \in X} N[v]$ and $N(X) =  N[X] \setminus X$, respectively.   We say that a vertex $v$ is \emph{simplicial} if $N[v]$ is a clique.

A graph $G$ is a \emph{threshold graph} if there exist a weight assignment $w : V(G) \rightarrow \mathbb{R}^+$ and a threshold $T \in \mathbb{R}^+$ such that $uv \in E(G)$ if and only if $w(u) + w(v) \ge T$.
Alternatively, a graph $G$ is a threshold graph if and only if its vertices can be partitioned into a clique and an independent set $I$ such that the neighborhoods of vertices in $I$ form a total order under the containment relation.
A graph $G$ is a \emph{paired threshold graph} if there exist a weight assignment $w : V(G) \rightarrow \mathbb{R}^+$ and two fixed thresholds $T_{+}, T_{-} \in \mathbb{R}^+$ such that $uv \in E(G)$ if and only if
\[
  w(u) + w(v) \geq T_{+} \quad\text{ and }\quad |w(u) - w(v)| \leq T_{-}.
\]
Given an assignment $w$ and thresholds $T_+, T_-$ for a paired threshold graph, we may adjust all the vertex weights by the same value $\epsilon$ and $T_+$ by $2\epsilon$, while keeping $T_-$ unchanged.  It is easy to verify that it represents the same graph.   Thus, we can always make the two thresholds equal.

\begin{proposition}
  A graph $G$ is a {paired threshold graph} if and only if there exist a weight assignment $w : V(G) \rightarrow \mathbb{R}^+$ and a threshold $T_{\pm} \in \mathbb{R}^+$ such that $uv \in E(G)$ if and only if
    $w(u) + w(v) \ge T_{\pm}$ and $|w(u) - w(v)| \le T_{\pm}$.
\end{proposition}
\begin{proof} 
  The if direction is trivial.  For the only if direction, suppose that $w'$ is the weight assignment, and $T'_+, T'_-$ are the thresholds in the definition.
  Note that if $w'(v) < (T'_+ - T'_-)/2$, $v$ is an isolated vertex in $G$.
  For these vertices, we reset $w'(v)$ to be a big positive value such that $w'(v) - w'(u) > T'_-$ for each $u \in V(G) \setminus \{v\}$.
  If $w'(v) = (T'_+ - T'_-)/2$, we reset $w'(v) = w'(v) + \varepsilon$, where $\varepsilon$ is a sufficiently small positive value.
  After the adjustment, $w'$ remains a valid weight assignment and $w'(v) > (T'_+ - T'_-)/2$ for each vertex of $G$.

  Then, for each $v \in V(G)$, we set $w(v) = w'(v) - T'_+/2 + T'_-/2$.
  Vertices $u$ and $v$ are adjacent in $G$ if and only if $T'_+ \le w'(u ) + w'(v) = w(u) + T'_+/2 - T'_-/2 + w(v) + T'_+/2 - T'_-/2 = w(u) + w(v) + T'_+ - T'_-$, i.e., $w(u) + w(v) \ge T'_-$ and $T'_- \ge |w'(u ) - w'(v)| = |(w(u) + T'_+/2 - T'_-/2) - (w(v) + T'_+/2 - T'_-/2)| = |w(u) - w(v)|$.  Setting $T_\pm = T'_-$ will complete the proof.
\end{proof}

In the rest of the paper we use the same value for both thresholds.  We use $G(w, T_\pm)$ to denote a paired threshold graph with weight assignment $w$ and threshold $T_\pm$. 

In a threshold graph $G$ with weight assignment $w$ and threshold $T$, the weight $T/2$ defines a natural partition of the vertices: $\{v\mid w(v) < T/2\}$ forms an independent set, while $\{v\mid w(v) \ge T/2\}$ a clique.  We can use $T_\pm/2$ to get a similar partition of the vertex set of a paired threshold graph $G(w, T_\pm)$.  In particular, $\{v \mid w(v) < T_\pm/2\}$ remains an independent set.  However, $\{v \mid w(v) \ge T_\pm/2\}$ is no longer a clique in general.  Since the weight sum of any two such vertices in this set is at least $T_\pm$, it induces an indifference graph, or more widely known, a unit interval graph \cite{roberts-69-indifference-graphs}.
A graph is a \emph{unit interval graph} if its vertices can be assigned to unit-length intervals on the real line such that two vertices are adjacent if and only if their corresponding intervals intersect.

\begin{proposition}\label{prop:unit-interval}
  In a paired threshold graph $G(w, T_\pm)$, the subgraph induced by $\{v \mid w(v) \ge T_\pm/2\}$ is a unit interval graph.
\end{proposition}

  As a matter of fact, all unit interval graphs are paired threshold graphs.  This has been observed in \cite{ravanmehr-18-paired-threshold-graphs}, and we include a proof because the construction used in it will be exemplary in this paper.
An interval model can be specified by the $2 n$ endpoints for the $n$ intervals: The interval for vertex $v$ is denoted by $[\lp{v}, \rp{v}]$, where \lp{v} and \rp{v} are the, respectively, left point and the right point of the interval.
\begin{proposition}\label{prop:uig-assignment}
  A unit interval graph $G$ is a paired threshold graph.  Moreover, there is an assignment $w$ such that $w(v) \ge T_\pm$ for all vertices $v\in V(G)$.
\end{proposition}
\begin{proof}
  We may start with any unit interval model for $G$.  We first scale the intervals such that all intervals have length $T_\pm$.  Then we increase all the endpoints by the same value such that $\min_{v\in V(G)} \lp{v} \ge T_\pm$.    It is easy to verify that setting $w(v) = \lp{v}$ for all vertices $v$ is a valid weight assignment for the graph.
\end{proof}

The discussion above can be summarized as that a paired threshold graph can be partitioned into an independent set and a unit interval graph.
Now that we have thus fully understood both parts, it is time to put their connection under scrutiny.  Although the following is straightforward, we want to point out that the assignment $w$ for vertices in $N[I]$ is not necessarily a threshold assignment for them with respect to threshold $T_\pm$.

\begin{proposition}\label{lem:intersection}
  Let $G(w, T_\pm)$ be a paired threshold graph, and let $I = \{v \mid w(v) < T_\pm / 2\}$. Then $N[I]$ induces a threshold graph.
\end{proposition}
\begin{proof}
  For any pair of vertices $y, z \in N(I)$ with $w(y) \le w(z)$, there must be a vertex $x \in I$ such that $w(x) < w(y) \le w(z)$ and $xz \in E(G)$.
  By definition, $w(x) + w(z) \ge T_\pm$ and $w(z) - w(x) \le T_\pm$. As a result,
  \begin{align*}
  w(y) + w(z) &> w(x) + w(z) \ge T_\pm
  \text{ and}
  \\
  w(z) - w(y) &< w(z) - w(x) \le T_\pm,
  \end{align*}
  which imply $yz \in E(G)$. Hence $N(I)$ composes a clique.
  
  It is known that $I$ is an independent set of $G$. To prove that $N[I]$ induces a threshold graph, it suffices to prove that for any pair of vertices $x, y \in I$ with $w(x) \le w(y)$, $N(x) \subseteq N(y)$.
  If $N(x) = \emptyset$, clearly, $N(x) \subseteq N(y)$.
  Otherwise, let $z$ be any vertex of $N(x)$. Then, $w(z) \ge T_\pm /2 > w(y)  \ge w(x)$. By definition, $w(x) + w(z) \ge T_\pm$ and $w(z) - w(x) \le T_\pm$. As a result,
  \begin{align*}
  w(y) + w(z) &\ge w(x) + w(z) \ge T_\pm
  \text{ and}
  \\
  w(z) - w(y) &\le w(z) - w(x) \le T_\pm,
  \end{align*}
  which imply $yz \in E(G)$. Hence $N(x) \subseteq N(y)$.
\end{proof}

However, Propositions~\ref{prop:unit-interval} and~\ref{lem:intersection} are not sufficient for a graph to be a paired threshold graph; e.g., for each of the two graphs in Figures~\ref{fig:smallgraph}, setting $I$ to be a single simplicial vertex, we get a partition satisfying both propositions.  There is a more technical condition on the connection in between.

\begin{figure}[h]
  \centering      
  \footnotesize
  \begin{subfigure}[b]{0.17\linewidth}
    \centering
    \begin{tikzpicture}[scale=.2]
      \node [vertex] (s) at (0,6) {};
      \node [vertex] (a) at (-5,0) {};
      \node [vertex] (a1) at (-2,0) {};
      \node [vertex] (b1) at (2,0) {};
      \node [vertex] (b) at (5,0) {};
      \node [vertex] (c) at (0,3.5) {};
      \draw[] (a) -- (a1) -- (b1) -- (b);
      \draw[] (c) -- (s);
      \draw[] (a1) -- (c) -- (b1);
    \end{tikzpicture}
    \caption{}\label{fig:net}    
  \end{subfigure} 
  \begin{subfigure}[b]{0.17\linewidth}
    \centering
    \begin{tikzpicture}[scale=.2]
      \node [vertex] (s) at (0,6) {};
      \node [vertex] (a) at (-4,0) {};
      \node [vertex] (a1) at (0, 0) {};
      \node [vertex] (b) at (4,0) {};
      \node [vertex] (c1) at (-2,3) {};
      \node [vertex] (c2) at (2,3) {};
      \draw[] (a) -- (a1) -- (b) -- (c2) -- (s) -- (c1) -- (a);
      \draw[] (c1) -- (c2) -- (a1) -- (c1);
    \end{tikzpicture}
    \caption{}\label{fig:tent}    
  \end{subfigure} 
  \caption{Two minimal non-paired threshold graphs.}
  \label{fig:smallgraph}
\end{figure}
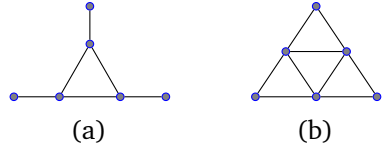

\section{Characterization}\label{section:Chara}
An \emph{ordering} $\sigma$ of the vertex set of a graph $G$ is a bijection from $V(G) \rightarrow \{1, \ldots, n\}$. We use $u <_\sigma v$ to denote that $\sigma(u) < \sigma(v)$.
An ordering $\sigma$ of the vertex set of a graph $G$ is an \emph{umbrella ordering} if for every triple of vertices $x,y,z$,
\[
  x <_\sigma y <_\sigma z\text{ and }xz \in E(G)\text{ imply }xy,yz \in E(G).
\]
Looges and Olariu~\cite{looges-93-greedy-algorithms-uig} showed that a graph is a unit interval graph if and only if it admits an {umbrella ordering}.
Indeed, given a unit interval model of a unit interval graph $G$, ordering the vertices by their left endpoints produces an umbrella ordering of $G$.
Likewise interval graphs can be characterized by the so called \emph{interval orderings} $\sigma$ \cite{ramalingam-88-domination-interval-graphs}: for every triple of vertices $x, y, z$, 
\[
x <_\sigma y <_\sigma z\text{ and }xz \in E(G)\text{ imply }xy \in E(G).
\]
We now formally define the broom ordering.

\begin{definition}\label{def:besom-ordering}
  An ordering $\sigma$ of the vertex set of a graph $G$ is a \emph{broom ordering} if its reversal is an interval ordering of $G$ and there exists $p$ with $0\le p\le n$ such that
  \begin{enumerate}[(i)]
  \item for each of the first $p$ vertices of $\sigma$, its neighborhood is after $p$ and appears consecutively in $\sigma$; and
  \item the sub-ordering of the last $n - p$ vertices is an umbrella ordering.
  \end{enumerate}
\end{definition}

Two remarks on this definition are in order.  For each of the last $n - p$ vertices, its closed neighborhood appears consecutively in $\sigma$.  A graph having a broom ordering has also an interval ordering, hence an interval graph.  Therefore, the following lemma implies Theorem~\ref{thm:PTisInterval}.

\begin{lemma}\label{lemma:pt-besom-ordering}
  Let $G(T_\pm, w)$ be a paired threshold graph. The ordering of $V(G)$ decided by $w$, with ties broken arbitrarily, is a broom ordering.
\end{lemma}
\begin{proof}
  Let $\sigma$ be the ordering.  We first show that its reversal is an interval ordering.   Let $x, y, z$ be three distinct vertices of $G$ with $x <_\sigma y <_\sigma z$; note that $w(x) \leq w(y) \leq w(z)$.  If $xz \in E(G)$, then by definition, $w(x) + w(z) \ge T_\pm$ and $w(z) - w(x) \le T_\pm$.
  As a result,
  \begin{align*}
    w(y) + w(z) &\ge w(x) + w(z) \ge T_\pm
                  \text{ and}
    \\
    w(z) - w(y) &\leq w(z) - w(x) \le T_\pm,
  \end{align*}
  which imply $yz \in E(G)$.  Therefore, the reversal of $\sigma$ is an interval ordering.
  
  Let $I = \{ v \mid w(v) < T_\pm / 2 \}$.
  If $I$ is empty, then $G$ is a unit interval graph and $\sigma$ is an umbrella ordering of $G$. It is clear that $\sigma$ is a broom ordering with $p = 0$.
  Otherwise, let $p = |I|$. There is an edge between a vertex $v\in I$ and a vertex $u\in V(G)\setminus I$ if and only if
  \[
    T_\pm - w(v) \le w(u) \le     T_\pm + w(v). 
  \]
  Then condition (i) and (ii) follow from the definition of $\sigma$ and Proposition~\ref{prop:unit-interval}, respectively.
\end{proof}

Before proving the main result of this section, we need to take care of disconnected graphs.

\begin{lemma}\label{lem:disconnected}
  A graph $G$ is a paired threshold graph if and only if one component is a connected paired threshold graph and all the others are unit interval graphs.
\end{lemma}
\begin{proof}
  For the if direction, let $G[C]$ be the connected paired threshold graph, and let $w, T_\pm$ be an assignment for $G[C]$.  Note that $G - C$ is a unit interval graph, for which we find a weight assignment $w'$ using Proposition~\ref{prop:uig-assignment}.  
  We set $w(v) = w'(v) + \max_{c\in C}w(c) + 1$ for each $v\in V(G)\setminus C$.   
  It is easy to verify that $w$ is a valid paired threshold assignment for $G$, and thus $G$ is a paired threshold graph.

  We now prove the only if direction.  Suppose for contradiction that there are two components, $G[C_1]$ and $G[C_2]$, that are not unit interval graphs.  We fix a weight assignment $w$ and threshold  $T_\pm$ of $G$.  By Proposition~\ref{prop:unit-interval}, there must be two vertices $u_1\in C_1$ and $u_2\in C_2$ such that $w(u_1), w(u_2)< T_\pm/2$.
  We can find a neighbor $v_1$ of $u_1$ and a neighbor $v_2$ of $u_2$; otherwise the component consists of a single vertex and is a unit interval graph.
  By Lemma~\ref{lemma:pt-besom-ordering}, the assignment $w$ decides a broom ordering $\sigma$ of $G$, in which, we may assume without loss of generality that  $v_1 <_\sigma v_2$.
  Then, we have that $u_1 <_\sigma u_2 <_\sigma v_1 <_\sigma v_2$ or $u_2 <_\sigma u_1 <_\sigma v_1 <_\sigma v_2$.
  Since the reversal of $\sigma$ is an interval ordering, at least one of $u_1 v_2$ and $u_2 v_1$ is an edge of $G$.  This contradicts that they (two ends of the edge) are from different components.  
\end{proof}

Another way to characterize paired threshold graphs is through partition, with the focus on the connection between two parts.
\begin{definition} \label{def:pt-partition}
  A partition $I\uplus U$ of a graph $G$ is a paired threshold partition of $G$ if 
  \begin{enumerate}[(i)]
  \item $I$ is an independent set, $G[U]$ is a unit interval graph, and $N[I]$ induces a threshold graph;
  \item $G[U]$ has an umbrella ordering $\sigma$ in which $N(v)$ appears consecutively for each $v\in I$; and
  \item  $N(I) \subseteq N[u]$, where $u$ is the first vertex of $\sigma$.
  \end{enumerate}
\end{definition}

\begin{figure}[h]
  \begin{tikzpicture}[scale=0.35]
    \small
    \draw[dashed] (-1, 0) -- (33, 0);

    \foreach \x in {0,1,...,12} {
      \node[vertex, fill=cyan] (u\x) at (\x, 0.5+\x/2) {};
      \draw[thin, gray, -{|[left]}] (u\x) to (\x+12.5, 0.5+\x/2);
    }
    \foreach[count=\i from 0] \x in {13,14,15,20} {
      \node[vertex, fill=cyan] (u\x) at (\x+1, 0.5+\i/2) {};
      \draw[thin, gray, -{|[left]}] (u\x) to (\x+13.5, 0.5+\i/2);
    }
    \draw[dashed,gray,thin] (12.5, 0) -- (12.5, 7.5);

    \node[vertex, black, label=above:$v_1$] (v_1) at (-5, 8) {};
    \foreach \x in {4,5,6}
    {\draw[thin, black, thick] (v_1) edge (u\x);}

    \node[vertex, violet, label=above:$v_2$] (v_i) at (-3.5, 8.5) {};
    \foreach \x in {4,...,8}
    {\draw[thin, violet] (v_i) edge (u\x);}

    \node[vertex, olive, "$v_3$"] (v_p) at (-2,8.9) {};
    \foreach \x in {3,...,10}
    {\draw[thin, cyan] (v_p) edge (u\x);}

    \draw[dashed, gray] (-3.,8.3) ellipse (3 and 2.); \node at (-4.,6.) {$I$};

    \node[left] at (u0) {$v_4$};
    \node[left] at (u3) {$v_7$};
    \node[left] at (u4) {$v_8$};
    \node[above left] at (u12) {$v_{16}$};
    \node[left] at (u20) {$v_{20}$};        
  \end{tikzpicture}
  \caption{Illustration for the proof of Theorem~\ref{thm:Characterization} $(3) \Rightarrow (1)$, where $p = 3$, $q = 8$, $t = 16$, $n = 20$, and $T_\pm = 2 * 20^2 = 800$.  The weights for the vertices derived from the proof are $20$, $40$, $60$, $724$, $725$, $726$, $747$, $ 808$, $809$, $810$, $831$, $832$, $853$, $854$, $ 875$, $876$, $1525.85$, $1544.9$, $1604.95$, and $1653$, in order.}
  \label{fig:example of thm}
\end{figure}
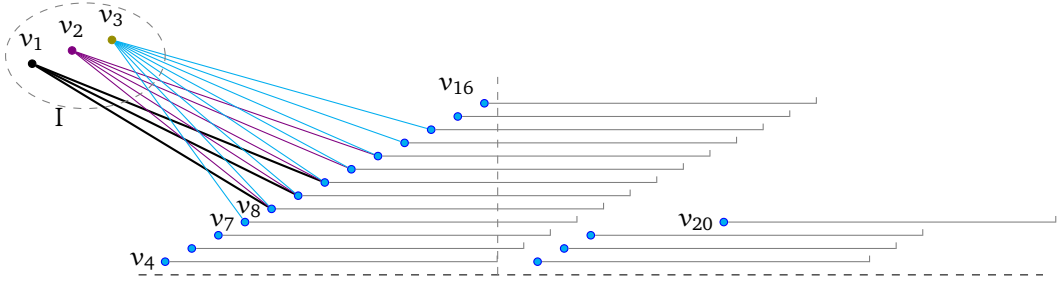

The following theorem implies Theorem~\ref{thm:broom-ordering-for-pt}.

\begin{theorem}\label{thm:Characterization}
  The following are equivalent on a graph $G$.
  \begin{enumerate}[(1)]
  \item $G$ is a paired threshold graph.
  \item $G$ admits a broom ordering.
  \item $G$ has a paired threshold partition.
  \end{enumerate}
\end{theorem}
\begin{proof}
  
  $(1) \Rightarrow (2)$ This has been proved in Lemma~\ref{lemma:pt-besom-ordering}.
  
  $(2) \Rightarrow (3)$ Let $\sigma$ be a broom ordering of $G$ with integer $p$.
  We take the first $p$ vertices in $\sigma$ as $I$ and the rest as $U$.  Note that $I$ is an independent set, and $G[U]$ is a unit interval graph.  Condition (ii) of Definition~\ref{def:pt-partition} follows directly from Definition~\ref{def:besom-ordering}: The sub-ordering of $\sigma$ restricted to $U$ is an umbrella ordering of $G[U]$, in which the neighborhood of each vertex of $I$ appears consecutively.  
To show that $N(I)$ is a clique, consider two vertices $x, y$ in $N(I)$ with $x <_\sigma y$.  Let $v \in N(y)\cap I$.  Clearly, $v <_\sigma x <_\sigma y$, which implies $x y \in E(G)$ because $\sigma$ is the reversal of an interval ordering. 
  We argue then that the neighborhoods of vertices in $I$ form a total order under the containment relation.  This is vacuous when $p \le 1$.   Now let $u, v$ be any two vertices in $I$ with $u <_\sigma v$.  For any $x \in N(u)$, we have $u <_\sigma v <_\sigma x$; since the reversal of $\sigma$ is an interval ordering of $G$, this implies  $x$ is in $N(v)$ as well. 
  Therefore, $N[I]$ induces a threshold graph and we have verified condition (i) of  Definition~\ref{def:pt-partition}.
  Condition (iii) of  Definition~\ref{def:pt-partition} is vacuously true when $U$ is empty.  Otherwise, let $x$ be the first vertex of $U$ in $\sigma$.  Then for any $v \in I$ and $y \in N(v)\setminus \{x\}$, we have $v <_\sigma x <_\sigma y$.  Then $y \in N[x]$ because $\sigma$ is the reversal of an interval ordering of $G$.  We have thus concluded that $I\uplus U$ is a paired threshold partition of $G$.
  
  $(3) \Rightarrow (1)$   Let $I\uplus U$ be a paired threshold partition of $G$.  We find a weight assignment $w$ for $G$ with threshold $T_\pm = 2n^2$.
  We may assume without loss of generality that $G$ is connected: Otherwise, by Lemma~\ref{lem:disconnected}, we focus on the only component of $G$ that is not a unit interval graph.  Note that $u$ belongs to this component because the component intersects both $I$ and $U$ and condition (iii) of Definition~\ref{def:pt-partition}.  Therefore, the partition restricted to the component remains a valid paired threshold partition.
  On the other hand, this direction holds vacuously when $I$ is empty (when $G$ is a unit interval graph) or when $U$ is empty (when $|I| = 1$).  We may hence assume that neither of $I$ and $U$ is empty.
  
  Let $p = |I|$; note that $0< p < n$.   We may number the vertices in $G$ such that $I = \{ v_1, \ldots, v_p \}$ and $N(v_1) \subseteq \cdots \subseteq N(v_p)$, and $\sigma = \langle v_{p+1}, v_{p+2}, ..., v_{n} \rangle$.   Let $q$ be the smallest index such that $v_q \in N(v_1)$, and let $t$ be the largest index such that $v_t\in N[v_{p + 1}]$; they ($q$ and $t$) exist because $G$ is connected.
  If $t = n$, then $U$ is a clique and $G$ is a threshold graph; hence a paired threshold graph.
  In the rest we assume $p+1 \le q \le t < n$.
  For each $i$ with $p+1 \le i \le n$, let $s(i)$ denote the smallest index such that $v_{s(i)} v_i\in E(G)$, which exists because $G$ is connected.
  We remark that if $v_i \in N(I)$, then $1 \le s(i) \le p$; and if $t < i \le n$, then $p+1 < s(i) < i$.
  Moreover, if $t < i < j \le n$, then $s(i) \le s(j)$.
  		
  For each $i$ with $1 \le i \le p$, we set $w(v_i) = i\cdot n$.  For each $i$ with $p+1 \le i \le t$, we set
  \begin{equation}
    \label{eq:w}
    w(v_i) =
    \begin{cases}
      (2n - p)n - ({n - i}) & \text{ if } i < q, v_i\not\in N(I),
      \\
      (2n - s(i))n + {i} & \text{ if } i < q, v_i\in N(I),
      \\
      (2n + s(i))n - ({n - i}) & \text{ if } i \ge q, v_i\in N(I),
      \\
      (2n + p)n + {i} & \text{ if } i > q, v_i\not\in N(I).
    \end{cases}
  \end{equation}
  Finally, for each $i$ with $t < i \leq n$, we set
  \begin{equation}
    \label{eq:w2}
    w(v_i) = 2n^2 + w(v_{s(i)-1}) + \frac{i}{n} \big(w(v_{s(i)}) - w(v_{s(i)-1})\big),
  \end{equation}
  which is well defined because $1 \le p < s(i) < i$. See Figure~\ref{fig:example of thm} for an example of the assignment.
  
  We claim that the weights are increasing with the vertex numbers.  First, $w(v_1) < \cdots < w(v_p) = p n < n^2$.  On the other hand, $w(v_{p + 1}) > n^2$ can be read directly from the weight assignment in \eqref{eq:w}.  Now let $i, j$ be two numbers with $p+1 \le i < j < q$; note that if $v_i\in N(I)$ then we must have $v_j\in N(I)$ as well by the selection of $q$ and the numbering of vertices.
  \begin{itemize}
  \item Case 1, $v_i, v_j \notin N(I)$.   We have $w(v_i) < w(v_j)$ because $w(v_j) - w(v_i) = \big( (2n-p)n - (n-j) \big) - \big( (2n-p)n - (n-i) \big) = j - i > 0$.
  \item  Case 2, $v_i \notin N(I)$ and $v_j \in N(I)$.  Noting that $s(j) \le p$, we have $w(v_j) - w(v_i) = \big( (2n - s(j))n + j \big) - \big( (2n - p)n - (n-i) \big) \ge n + j - i > 0$, i.e., $w(v_i) < w(v_j)$.
  \item  Case 3, $v_i, v_j \in N(I)$.  Note that ${s(i)}, {s(j)} \le p$.  Then $v_q\in N(v_1) \subseteq N(v_{s(i)})$ because the selection of $q$ and the numbering of the vertices respectively.  Now that $v_i, v_q \in N(v_{s(i)})$, we have that $v_j$ is in $N(v_{s(i)})$ as well because $i < j < q$ and  $N(v_{s(i)})$ appears consecutively in $\sigma$.  As a result, $s(j) \le s(i)$ and $w(v_j) - w(v_i) = \big( (2n -s(j))n + j \big) - \big( (2n - s(i))n + i \big) = (s(i) - s(j))n + j - i > 0$.
  \end{itemize}
  For any $p< i < q$, it is directly from \eqref{eq:w} that $w(v_i) < 2 n^2 < w(v_q)$.  Now suppose $q \le i < j \le t$; note that if $v_i\not\in N(I)$ then we must have $v_j\not\in N(I)$ as well by the selection of $q$ and the numbering of vertices. 
  \begin{itemize}
  \item Case 1, $v_i, v_j \in N(I)$.  Note that ${s(i)}, {s(j)} \le p$.  Then $v_q\in N(v_1) \subseteq N(v_{s(j)})$ because the selection of $q$ and the numbering of the vertices respectively.  Now that $v_j, v_q \in N(v_{s(j)})$, we have $v_i$ is in $N(v_{s(j)})$ as well because $q \le i < j$ and  $N(v_{s(j)})$ appears consecutively in $\sigma$.  As a result, $s(i) \le s(j)$ and $w(v_j) - w(v_i) = \big( (2n + s(j))n - (n - j) \big) - \big( (2n + s(i))n - (n - i) \big) = (s(j) - s(i))n + j - i > 0$.
  \item Case 2, $v_i \in N(I)$ and $v_j \notin N(I)$.  Noting that $s(i) \le p$, we have $w(v_j) - w(v_i) = \big( (2n + p)n + j\big) - \big( (2n + s(i))n - (n - i) \big) \ge j + n - i > 0$.
  \item Case 3, $v_i, v_j \notin N(I)$.
  Then, $w(v_j) - w(v_i) = \big( (2n + p)n + j \big) - \big( (2n + p)n + i \big) = j - i > 0$. 
  \end{itemize}
  Since $ p+1 \le t < n$, it follows directly from \eqref{eq:w} that $w(v_t) < 3n^2$, whether $v_t \in N(I)$ or not.  On the other hand, $s(t+1) > p+1$ by the selection of $t$, and hence $w(v_{t+1}) = 2n^2 + w(v_{s(t+1)-1}) + \frac{t+1}{n} \big(w(v_{s(t+1)}) - w(v_{s(t+1)-1})\big) > 2n^2 + w(v_{p+1}) > 3n^2$.  
 For the last part, i.e., $t+1, \ldots, n$, we show by induction that $w(v_{i + 1}) > w(v_i)$ for all $i$ with $t < i < n$.  By the selection of $t$, we have $s(i), s(i+1) > p+1$.  On the other hand, $s(i) \le s(i+1) \le i$ because $\sigma$ is an umbrella ordering of $G[U]$ and it is connected.  
 Thus, $w(v_{i+1}) - w(v_i) = \frac{n-({i+1})}{n} w(v_{s({i+1})-1}) + \frac{i+1}{n} w(v_{s({i+1})}) - \frac{n-i}{n} w(v_{s(i)-1}) - \frac{i}{n} w(v_{s(i)}) \ge \frac{n-({i+1})}{n} w(v_{s({i})-1}) + \frac{i+1}{n} w(v_{s({i})}) - \frac{n-i}{n} w(v_{s(i)-1}) - \frac{i}{n} w(v_{s(i)}) = \frac{1}{n} w(v_{s(i)}) - \frac{1}{n} w(v_{s(i)-1}) > 0$.

  Putting them together, we have
  \begin{equation}
    \tag{$\star$}
    \label{eq:weight-order}
    w(v_1) < \cdots < w(v_p) < n^2 < w(v_{p + 1}) < \cdots w(v_t) < 3 n^2 < w(v_{t+1}) < \cdots < w(v_n).
  \end{equation}
   
  It remains to verify that $w$ is a valid assignment for $G$---i.e., for each pair of $i,j$ with $1\le i< j\le n$,
  \[
    v_i v_j \in E(G) \quad\text{if and only if}\quad  w(v_j) - w(v_i) \le T_\pm\le w(v_j) + w(v_i).
  \]

  Case 1, $i< j\le p$ (i.e., $v_i, v_j \in I$).  In this case $v_iv_j \notin E(G)$ and $w(v_i) + w(v_j) < 2n^2$.  

  Case 2, $i\le p< j$ (i.e., $v_i \in I$ and  $v_j \in U$).  We break into five sub-cases; the first four of them correspond to the assignment in \eqref{eq:w}.
  \begin{itemize}
  \item Case 2.1, $j < q$ and $v_j \notin N(I)$.  In this case, $v_iv_j \notin E(G)$ and
    $
    w(v_i) + w(v_j) = i n + (2n-p)n - (n-j) = 2 n^2 - (p - i) n - (n - j) < 2n^2.
    $
  \item Case 2.2, $j < q$ and $v_j \in N(I)$.  From the definition of $s(j)$ and the assumption $N(v_1) \subseteq \cdots \subseteq N(v_p) $, we can infer that $v_iv_j \in E(G)$ if and only if $s(j) \leq i$, which holds if and only if $w(v_i) + w(v_j) = i \cdot n + (2n-s(j))n + j  = 2 n^2 + (i - s(j)) n + j> 2n^2$.  On the other hand, we always have $w(v_j) - w(v_i) = (2n-s(j))n + j - i \cdot n < 2 n^2$ in this case.
  \item Case 2.3, $q \le j\le t$ and $v_j \in N(I)$.  Similar as case 2.2,  $v_iv_j \in E(G)$ if and only if $s(j) \leq i$, which holds if and only if $w(v_j) - w(v_i) =  (2n+s(j))n - (n-j) - i \cdot n = 2 n^2 - (i - s(j)) n - (n - j)< 2n^2$.  Note that we always have $w(v_i) + w(v_j) = i \cdot n + (2n+s(j))n - (n-j) > 2n^2$.
  \item Case 2.4, $q < j \le t$ and $v_j \notin N(I)$.  In this case, $v_iv_j \notin E(G)$ and $w(v_j) - w(v_i) = (2n+p)n + j - i \cdot n = 2n^2 + (p - i) n + j > 2n^2$.
  \item Case 2.5, $t < j \le n$.  By condition (3), $v_j \notin N(I)$, which means $v_i v_j\not\in E(G)$.  Then $s(j) > p + 1$ and $w(v_j) - w(v_i) > 2n^2 + w(v_{p+1}) - w(v_i) > 2n^2$, by \eqref{eq:weight-order}.
  \end{itemize}

  Case 3, $p< i< j \le n$ (i.e., $v_i, v_j \in U$).  Note that in this case $v_i v_j\in E(G)$ if and only if $s(j) \le i$.  Since it always holds that $w(v_i) + w(v_j) > 2n^2$, here we focus on the difference $w(v_j) - w(v_i)$.  Note that $v_iv_j \in E(G)$ when $j \le t$, and then by \eqref{eq:weight-order}, $w(v_j) - w(v_i) < 2n^2$.  In the rest, $t< j$.  If $v_i v_j\not\in E(G)$, i.e., $i < s(j)$, then $w(v_j) - w(v_i) = 2n^2+ w(v_{s(j)-1}) + \frac{j}{n} \big(w(v_{s(j)}) - w(v_{s(j)-1})\big) - w(v_{i}) \ge 2n^2+ w(v_{s(j)-1}) + \frac{j}{n} \big(w(v_{s(j)}) - w(v_{s(j)-1})\big) - w(v_{s(j) - 1}) > 2n^2$.  Otherwise,  $v_i v_j\in E(G)$, then $s(j)\le i$, and $w(v_j) - w(v_i) = 2n^2+ w(v_{s(j)-1}) + \frac{j}{n} \big(w(v_{s(j)}) - w(v_{s(j)-1})\big) - w(v_{i}) \le 2n^2 - \big( w(v_{s(j)}) - w(v_{s(j)-1}) \big) + \frac{j}{n} \big(w(v_{s(j)}) - w(v_{s(j)-1})\big) = 2n^2 - \frac{n - j}{n} \big(w(v_{s(j)}) - w(v_{s(j)-1})\big) \le 2n^2$.

Thus, $w$ is a valid assignment for $G$, and this completes the proof.
\end{proof}

The paired threshold partition $I \uplus U$ of a paired threshold graph is not necessarily unique: For example, if the first vertex in $U$ is nonadjacent to any vertex in $I$, then we may move it from $U$ to $I$ to make a new partition.  We say that a paired threshold partition  $I \uplus U$ is a \emph{canonical partition} if $I$ is maximal in all partitions,---i.e., $(I \cup \{v\}) \uplus (U \setminus \{v\})$ is not a valid paired threshold partition for any $v\in U$.  The following characterizes {canonical partitions}.

\begin{proposition}\label{prop:CanonicalProperty}
  Let $I \uplus U$ be a canonical partition of a paired threshold graph $G$.  If a vertex $v \in U\setminus N(I)$ is simplicial in $G$, then
  \begin{enumerate}[(1)]
  \item $N(I) \nsubseteq N(v)$; and
  \item the subgraph induced by $N[I \cup \{v\}]$ is not a threshold graph.
  \end{enumerate}
\end{proposition}
\begin{proof}
  Let $\sigma$ be an umbrella ordering of $G[U]$ specified in Definition~\ref{def:pt-partition}, and let $u$ be the first vertex of $\sigma$.  Note that $u\in N(I)$; otherwise, we can produce another paired threshold partition $(I \cup \{u\}) \uplus ( U \setminus \{u\})$, contradicting that $I \uplus U$ is a canonical partition.  Hence $u \ne v$.  We show that if either of (1) and (2) is not true, then $I' \uplus U'$, where $I' = I \cup \{v\}$ and $U' = U \setminus \{v\}$, is also a valid paired threshold partition.  Note that $I'$ is an independent set of $G$.  Let $\sigma'$ be the ordering obtained from $\sigma$ by removing $v$.  By the definition of umbrella orderings, vertices in $N(v)$ appear consecutively in $\sigma'$.  

(1) Suppose that $N(I) \subseteq N(v)$. 
Note that $u \in N(v)$ because $u \in N(I)$.
Since $v$ is a simplicial vertex, $N[v]$ is a maximal clique of $G$; it contains $u$, hence $N[v] \subseteq N[u]$.  Therefore, partition $I' \uplus U'$ satisfies conditions of Definition~\ref{def:pt-partition}.

(2) Suppose that $N[I']$ induces a threshold graph.  We are already in the previous case if $N(I) \subseteq N(v)$; hence we assume $N(I) \nsubseteq N(v)$.  Then by the definition of threshold graphs, we must have $N(v) \subset N(I)$, and thus $N(I') = N(I) \subseteq N[u]$.
Therefore, partition $I' \uplus U'$ satisfies conditions of Definition~\ref{def:pt-partition}.
\end{proof}

If we drop the unit-length from the definition of unit interval graphs, then we end with {interval graphs},---i.e., it allows intervals of arbitrary lengths.
A graph $G$ is an interval graph if and only if its maximal cliques can be arranged in a sequence such that for every vertex $v \in V(G)$, the set of nodes containing $v$ occur consecutively in the sequence \cite{gilmore-64-comparability-and-interval}.  This sequence is called a \emph{clique path} of $G$.  It is known that a connected unit interval graph has a unique clique path, up to full reversal \cite{deng-96-proper-interval-and-cag,corneil-04-recognize-uig}. 

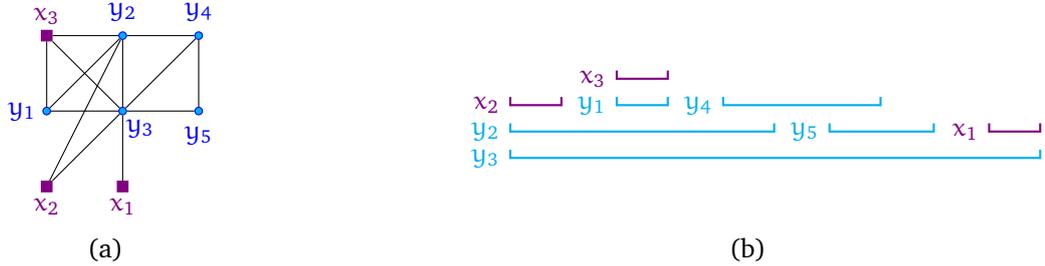
\begin{figure}[h]
  \centering\small
\begin{subfigure}[b]{0.35\linewidth}
    \centering
  \begin{tikzpicture}[, every path/.style={}]
    \begin{scope}[every node/.style={uvertex}]
    \node[fill=violet, "$x_1$" below](x1) at (1,0) {};
    \node[fill=violet, "$x_2$" below](x2) at (0,0) {};
    \node[fill=violet, "$x_3$" above](x3) at (0,2) {};
    \end{scope}
      
    \begin{scope}[every node/.style={vertex}]
    \node[fill=cyan, "$y_1$" left] (y1) at (0,1) {};
    \node[fill=cyan, "$y_2$" above] (y2) at (1,2) {};
    \node[fill=cyan, "$y_4$" above] (y3) at (2,2) {};
    \node[fill=cyan, "$y_5$" below] (y4) at (2,1) {};
    \node[fill=cyan, "$y_3$" below right] (y5) at (1,1) {};
    \end{scope}

    \draw (x1) -- (y5) (y2) -- (x2) -- (y5);
    \draw (y1) -- (x3) -- (y2) (x3) -- (y5);

    \draw (y1) -- (y5) -- (y2) (y3) -- (y5) -- (y4);
    \draw (y1) -- (y2) -- (y3) -- (y4);
  \end{tikzpicture}
    \caption{\empty}
  \end{subfigure} 
  \qquad
  \begin{subfigure}[b]{0.55\linewidth}
    \centering
  \begin{tikzpicture}[scale=.7]
  \draw[{|[right]}-{|[left]}, violet, thick] (4,2) node[left] {$x_3$} to (5,2);
  \draw[{|[right]}-{|[left]}, violet, thick]  (2,1.5) node[left] {$x_2$} to (3,1.5);
  \draw[{|[right]}-{|[left]}, violet, thick]  (11,1.) node[left] {$x_1$}  to (12,1.0);
  \draw[{|[right]}-{|[left]}, cyan, thick]  (2,1.0) node[left] {$y_2$} to (7,1.0);
  \draw[{|[right]}-{|[left]}, cyan, thick]  (8,1.0) node[left] {$y_5$} to (10,1.0);
  \draw[{|[right]}-{|[left]}, cyan, thick]  (4,1.5) node[left] {$y_1$} to (5,1.5);
  \draw[{|[right]}-{|[left]}, cyan, thick]  (6,1.5) node[left] {$y_4$} to (9,1.5);
  \draw[{|[right]}-{|[left]}, cyan, thick]  (2,0.5) node[left] {$y_3$} to (12,0.5);
  \node at (4, -.5) {};
  \end{tikzpicture}
    \caption{\empty}
  \end{subfigure} 
  \caption{ Illustration for the proof of Theorem~\ref{thm:CharacOfCliquePath}. (a)  A connected paired threshold graph and (b) its interval mode.
 A canonical partition of $G$ is given by $I = \{x_1,x_2,x_3\}$ and $U = \{y_1, y_2, y_3, y_4, y_5\}$.  The graph has five maximal cliques, $\{ x_2,y_2,y_3 \}$, $\{ x_3,y_1,y_2,y_3 \}$, $\{ y_2,y_3,y_4 \}$, $\{ y_3,y_4,y_5 \}$, $\{ x_1,y_3 \}$, which form a clique path of $G$ in order.  }
  \label{fig:interval-model}
\end{figure}

By Theorem~\ref{thm:PTisInterval}, we know that all paired threshold graphs are interval graphs. 
This leads us to consider paired threshold graphs on the perspective of clique path.
A paired threshold graph admits a clique path with some important properties to be used in our recognition algorithm.

\begin{theorem}\label{thm:CharacOfCliquePath}
  Let $G$ be a connected paired threshold graph, and let $I \uplus U$ be a canonical partition of $G$.  
  \begin{enumerate}[(1)]
  \item   There exists a clique path $C_1, \ldots, C_\ell$ of $G$ such that
    for $1\le i \le |I|$, the only vertex in $I\cap C_i$ is a simplicial vertex of $G$, and $C_i\setminus I\subseteq C_{|I|}$;
  \item In any clique path of $G$, the maximal cliques disjoint from $I$ appear consecutively, and they appear in the same order, up to full reversal;
  \item
    If $G$ remains connected after all universal vertices removed, then for any clique path of $G$, vertices in $I$ appear in the first $|I|$ or the last $|I|$ cliques. 
  \end{enumerate}
  \end{theorem}
\begin{proof}
  Let $p = |I|$.   It is easy to see that a unit interval graph has a canonical partition with $I \ne \emptyset$.  Therefore, we always have $p > 0$.  We find a threshold $T_\pm$ and a weight assignment $w$ as specified in Theorem \ref{thm:Characterization}; we may number the vertices of $G$ such that
  \[
    w(v_1) < \cdots < w(v_p) < T_\pm/2 < w(v_{p + 1}) < \cdots < w(v_n).
  \]
  Let $K_{p+1}, \ldots, K_q$ be the clique path of $G[U]$ implied by the umbrella ordering $v_{p + 1}, \ldots, v_n$ of $U$.
  By Definition~\ref{def:pt-partition}, $N(I) \subseteq K_{p+1}$.  Let $K_i = N[v_i]$ for $1 \leq i \leq p$.  We choose ${\cal P} = K_1, \ldots, K_p, K_{p + 2}, \ldots, K_q$ if $K_{p+1} = N(v_{p})$, or $K_1, \ldots, K_p, K_{p+1}, \ldots, K_q$ otherwise.

  We now verify that $\cal P$ is a clique path of $G$.  For $1\le i\le p$, since $v_i$ is a simplicial vertex in $G$, clique $K_i$ is a maximal clique of $G$, and the unique maximal clique containing $v_i$.   
  By Definition~\ref{def:pt-partition}, each clique $K_i$, $i \ge p + 2$, contains a vertex nonadjacent to $u$, and thus nonadjacent to $I$, hence maximal as well.
  Therefore, $\cal P$ contains all the maximal cliques of $G$ and each appears once.
  For each $v \in N(I)$, let $\ell(v)$ denote the smallest number such that $v_{\ell(v)} \in N(v)$, then $v\in K_j$ for all ${\ell(v)}\le j \le p$.  Therefore, the maximal cliques containing $v$ appear consecutively on $\cal P$.  This is trivial for other vertices: A vertex in $I$ is in a single maximal clique; and for each vertex in $U\setminus N(I)$, the condition is satisfied because $K_{p+1}, \ldots, K_q$ is a clique path of $G[U]$. Therefore, $\cal P$ is a clique path of $G$, and it satisfies (1).

  (2)  We may assume $q > p+1$; otherwise this assertion holds vacuously.  Let $\cal P'$ be an arbitrary clique path of $G$.
  Note that $G[U]$ is connected, and hence in any clique path containing the maximal cliques $K_{p+1}, \ldots, K_q$, they are in the order of $K_{p+1}, \ldots, K_q$ or its reversal.  
  We may assume without loss of generality, they appear in $\cal P'$ in the order of $K_{p+2}, \ldots, K_q$ if $K_{p+1} = N(v_{p})$, or $K_{p+1}, \ldots, K_q$ otherwise.  Suppose for contradiction that they are not consecutive in $\cal P'$; i.e., there are $1 \leq i \leq p$ and $j > p$ such that $K_i$ appears in between $K_j$ and $K_{j + 1}$.
We show that $G[U]$ has a claw (the four-vertex tree with three leaves), which is impossible because $G[U]$ is a unit interval graph \cite{wegner-67-dissertation}.  Note that $K_{j} \cap K_{j+1} \ne \emptyset$ because $G[U]$ is connected.
  
Case 1, $N(v_p) = K_p \setminus \{v_p\} \subset K_{p+1}$.
Since $(I, U)$ is a canonical partition, by Proposition \ref{prop:CanonicalProperty}, no vertex of $K_{p+1} \setminus K_p$ is simplicial in $G$.  Thus, $K_{p+1} \setminus K_p \subset K_{p+2}$, and $K_{p+1} \cap K_{p+2} \nsubseteq K_p$. 
Since $K_i\setminus \{v_i\}\subseteq K_p$ for each $1\le i\le p$, we can conclude that $j > p+1$.  Note that  $K_{j} \cap K_{j+1} \subseteq K_i \setminus \{v_i\} \subseteq K_p \setminus \{v_p\} \subset K_{p+1}$, which means $K_{j} \setminus (K_{p+1} \cup K_{j+1}) = K_{j} \setminus K_{p+1} \neq \emptyset$.  But then we can find a vertex $x_1 \in K_{j} \setminus (K_{p+1} \cup K_{j+1})$, a vertex $x_2 \in K_{j+1} \setminus K_{j}$, a vertex $x_3 \in K_{p+1} \setminus K_{j}$, and a vertex $x_0 \in K_{p+1} \cap K_{j+1}$, which induce a claw with center $x_0$.  

Case 2, $N(v_{p}) = K_{p+1}$, then $K_p \setminus \{v_p\} = N[v_{p+1}] \cap U$.
Since $K_{j} \cap K_{j+1} \subseteq K_i \setminus \{v_i\} \subseteq K_p$, none of the sets $K_{j} \setminus (K_{p} \cup K_{j+1})$ and $K_{j+1} \setminus (K_{p} \cup K_{j})$ and $K_{p} \cap K_{j} \cap K_{j+1}$ is empty.  We can find one vertex from each of them, which, together with $v_{p+1}$, induce a claw in $G[U]$.

(3)  This assertion holds vacuously when $p = 1$ or $p = q$; hence we may assume $1 < p < q$.  According to assertion~(2), it suffices to show that one of $C_1$ and $C_q$ is disjoint from $I$.
  Suppose for contradiction that there are two vertices $u\in K'_1\cap I$ and $v\in K'_q \cap I$, i.e., $K'_1 = N[u]$ and $K'_q = N[v]$.  We may assume without loss of generality $N(u) \subseteq N(v)$.  By the definition of clique paths, $N(u)$ is a subset of all maximal cliques.  But then all vertices in $N(u)$ are universal, and $v_1$ is isolated in $G - N(u)$, contradicting the assumption.  
\end{proof}

We say that a set of $k$ intervals $[\lp{v_1}, \rp{v_1}]$, $\ldots,$ $[\lp{v_k}, \rp{v_k}]$ is nested if
\[
  \lp{v_k} <\cdots < \lp{v_1} \le \rp{v_{1}} < \cdots < \rp{v_{k}}.
\]
\begin{figure}[h]
  \centering\small
  \begin{tikzpicture}[scale=.6]
  \foreach \ly in {1,2,...,4}
  {
    \draw[{|[right]}-{|[left]}, cyan, thin] (-\ly+5,\ly/3) node[black, left] {$v_\ly$} to (\ly + 8.5,\ly/3);
  }
  \foreach \ly in {5,6,...,9}
  {
    \draw[{|[right]}-{|[left]}, cyan, thin] (\ly,\ly/3) to (\ly + 8.5,\ly/3);
  }
  \foreach \ly in {10,11,...,16}
  {
    \draw[{|[right]}-{|[left]}, cyan, thin] (\ly,\ly/3-3) to (\ly + 8.5,\ly/3-3);
  }
  \foreach \lx in {1,2,...,4}
  {
    \draw[{|[right]}-{|[left]}, violet, thin] (-\lx+5,5/3) node[black,above] {$u_\lx$} to (-\lx + 5.5,5/3);
  }

  \draw[{|[right]}-{|[left]}, cyan, thick] (10, .333) to (18.5, .3333);  
  \node[above] at (5,5/3) {$v_5$};
  \node[left] at (9,3) {$v_9$};
  \node[right] at (18.5,1/3) {$v_{10}$};
  \node[right] at (24.5,7/3) {$v_{16}$};

  \draw[dashed,thin] (10.5, 3.3) -- (10.5, 0) node[below left, xshift=1mm] {\rp{v_2}};
  \draw[dashed,thin] (11.5, 3.3) -- (11.5, 0) node[below right, xshift=-1mm] {\rp{v_3}};
  \draw[dashed] (0, 0) -- (25, 0);
  \end{tikzpicture}
  \caption{Construction used in the proof of Lemma~\ref{thm:IntervalsOfPT}, with $k=4$.  In any interval model for this graph, we have to use intervals of different lengths for $v_1$, $v_2$, $v_3$, and $v_4$.}
  \label{fig:DifferenceIntervals}
\end{figure}
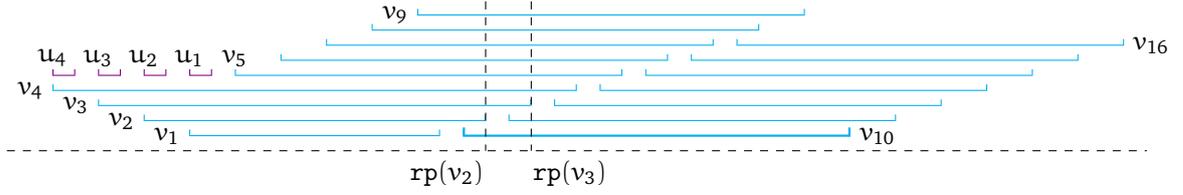

\begin{lemma}\label{thm:IntervalsOfPT}
  For any positive integer $k$, there exists a paired threshold graph $G$ such that there are $k$ nested intervals in  any interval model of $G$.
\end{lemma}
\begin{proof}
  We construct a graph $G$ as follows.  The vertex set consists $5 k$ vertices, denoted by $\{ v_1, \ldots, v_k, u_1, \ldots, u_{4k} \}$.  For each pair of $i, j$ with $1 \le i \le j \leq k$, we add edge $v_iu_{j}$.  For each $1 \leq i \leq 2k$, we add edges $u_i u_{i + 1}$, $\ldots$, $u_i u_{i + 2 k}$, i.e., connecting $u_i$ and the next $2 k$ vertices.  Finally, add all possible edges to make $u_{2 k+1}$, $\ldots$, $u_{4 k}$ into a clique.  In summary, the edge set of $G$ is
  \[
    E(G) = \{ u_iu_j \mid |j - i| \le 2k \}
    \cup  \{ v_iu_j \mid 1 \le i \le j \leq k \}.
  \]
  See Figure \ref{fig:DifferenceIntervals} for an example of the construction.
  It is easy to verify that $v_1, \ldots, v_k, u_1, \ldots, u_{4k}$ is a broom ordering, and hence $G$ is a paired threshold graph by Theorem~\ref{thm:Characterization}.
  
  We now argue that for any pair of $i, j$ with $1 \leq i < j \leq k$, the interval for $u_i$ has to be properly contained in the interval for $u_j$ in any interval model of $G$.  Note that $u_i$ and $u_{j+2k}$ are nonadjacent.  We may assume without loss of generality  $\rp{u_{i}} < \lp{u_{j+2k}}$.
Since both $u_{i}$ and $u_{j+2k}$ are adjacent to $u_j$ but they are not adjacent to each other,
 $$\lp{u_j}\le \rp{u_{i}} < \lp{u_{j+2k}} \leq \rp{u_{j}}.$$
 Vertex $u_{j+2k+1}$ is adjacent to $u_{j+2k}$ but not $u_j$, hence
 $$\rp{u_{j}} < \lp{u_{j+2k+1}} \le \rp{u_{j+2k}}.$$
 Since $u_{i+2k}$ is adjacent to both $u_{i}$ and $u_{j+2k+1}$, we have
 $$\lp{u_{i+2k}} \le \rp{u_{i}} < \lp{u_{j+2k+1}} \le \rp{u_{i+2k}}.$$
 From them we can conclude
 $$\lp{u_{i+2k}} \le \rp{u_{i}}, \rp{u_{j}} < \rp{u_{i+2k}},$$
 Finally, since $v_j$ is adjacent to $u_j$, but not to $u_i$ or $u_{i+2k}$,
 $$\lp{u_j} \le \rp{v_j} < \lp{u_i}.$$
 Putting them together, we can conclude that $\lp{u_j} < \lp{u_i} < \rp{u_{i}} < \rp{u_{j}}$, which completes the proof.
\end{proof}

\section{Recognition}\label{section:Recog}

It is well known that interval graphs, unit interval graphs, threshold graphs can be recognized in linear time \cite{hsu-99-recognizing-interval-graphs, corneil-04-recognize-uig}.  For (unit) interval graphs, the recognition algorithms return an interval model, from which we can retrieve a clique path or an umbrella ordering in the same time.
We say that two vertices $u,v$ are \emph{true twins} if $N[u] = N[v]$.  A set of vertices is a \emph{true-twin class} if it is a maximal set of vertices that are pairwise true twins.  A graph has a unique partition into true-twin classes.
The following proposition is from Deng et al.~\cite{deng-96-proper-interval-and-cag}, and it is the core idea of Corneil~\cite{corneil-04-recognize-uig}.

\begin{proposition}\label{prop:UIGSubOrderGeneralize}
  Let $G$ be a unit interval graph and let $\sigma$ be an umbrella ordering of $G$.
  \begin{enumerate}[(1)]
  \item For each set of true-twin class $T$ of $G$, vertices in $T$ appear consecutively in $\sigma$, and this subsequence can be replaced by an arbitrary ordering of $T$.
  \item If $G$ does not contain true twins, then it has a unique umbrella ordering, up to full reversal.
  \end{enumerate}
\end{proposition}

We may assume that the input graph $G$ is a connected interval graph.  If it is not an interval graph, then we may return ``no'' by Theorem \ref{thm:PTisInterval}.  If it is not connected, then we may remove all the components that are unit interval graphs, and return ``no'' if more than one component is left by Lemma \ref{lem:disconnected}.

The way we handle a connected interval graph $G$ is to use a clique path $\cal P$ of $G$.  We try to find a canonical partition from $\cal P$.  According to Theorem~\ref{thm:CharacOfCliquePath}(2), if $G$ is a paired threshold graph, then we can find the partition $I \uplus U$ with $I$ from the two ends of the clique path.  This can be simplified by Theorem~\ref{thm:CharacOfCliquePath}(3): If $G$ remains connected after all universal vertices removed, then it suffices to search only one end of $\cal P$.
We hence proceed dependent upon whether $G$ contains a universal vertex.

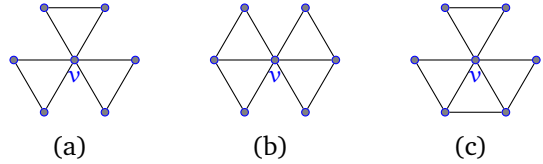
\begin{figure}[h!]
  \centering      
  \footnotesize
  \begin{subfigure}[b]{0.15\linewidth}
    \centering
    \begin{tikzpicture}[, scale=.8]
      \begin{scope}[every node/.style={vertex}]
      \node (a) at (60*4:1) {};
      \node (b) at (60*3:1) {};
      \node (c) at (60*2:1) {};
      \node (d) at (60*1:1) {};
      \node (e) at (60*0:1) {};
      \node (f) at (60*-1:1) {};
      \node["$v$" below] (v) at (0,0) {};
      \end{scope}

      \begin{scope}[every path/.style={}]
        \draw  (a)--(b)--(v)--(a) (c)--(d)--(v)--(c) (e)--(f)--(v)--(e);
      \end{scope}
    \end{tikzpicture}
    \caption{\empty}
    \label{fig:v+3K2}
  \end{subfigure} 
  \begin{subfigure}[b]{0.15\linewidth}
    \centering
    \begin{tikzpicture}[every node/.style={vertex}, scale=.8]
      \node (a) at (60*4:1) {};
      \node (b) at (60*3:1) {};
      \node (c) at (60*2:1) {};
      \node (d) at (60*1:1) {};
      \node (e) at (60*0:1) {};
      \node (f) at (60*-1:1) {};
      \node[label=below:$v$] (v) at (0,0) {};

      \begin{scope}[every path/.style={}]
        \draw  (a)--(b)--(c) (d)--(e)--(f) 
        (v)--(a) (v)--(b) (v)--(c) (v)--(d) (v)--(e) (v)--(f);
      \end{scope}
    \end{tikzpicture}
    \caption{\empty}
    \label{fig:v+2P3}
  \end{subfigure}
  \begin{subfigure}[b]{0.15\linewidth}
    \centering
    \begin{tikzpicture}[every node/.style={vertex}, scale=.8]
      \node (a) at (60*4:1) {};
      \node (b) at (60*3:1) {};
      \node (c) at (60*2:1) {};
      \node (d) at (60*1:1) {};
      \node (e) at (60*0:1) {};
      \node (f) at (60*-1:1) {};
      \node[label=below:$v$] (v) at (0,0) {};

      \begin{scope}[every path/.style={}]
        \draw  (b)--(a)--(f)--(e) (c)--(d) 
        (v)--(a) (v)--(b) (v)--(c) (v)--(d) (v)--(e) (v)--(f); 
      \end{scope}
    \end{tikzpicture}
    \caption{\empty}
    \label{fig:v+P2P4}
  \end{subfigure}
  \caption{Minimal non-paired threshold graphs each having a universal vertex.  Note that the removal of the universal vertex leaves a unit interval graph.}
  \label{fig:forbidden-with-universal}
\end{figure}

\begin{proposition}[\cite{chvatal-77-inequalities-in-ip}]
 \label{lem:deletion-sequence}
 A nonempty threshold graph contains either an isolated vertex or a universal vertex.
\end{proposition}

In other words, a threshold graph can always be made empty by exhaustively removing isolated vertices and universal vertices.  
It is easy to see that paired threshold graphs are closed under adding isolated vertices, but not necessarily universal vertices.  The following result is an extension of Proposition~\ref{lem:deletion-sequence} to paired threshold graphs.
Recall that minimal forbidden induced subgraphs of threshold graphs are $2K_2$, $P_4$, and $C_4$ (Figure~\ref{fig:threshold-free}).

\begin{lemma} \label{lemma:ComponentsOfPT}
  Let $G$ be a connected graph, and let $G'$ be obtained from $G$ by exhaustively removing universal vertices and isolated vertices.  If $G' \ne G$, then $G$ is a paired threshold graph if and only if one of the following conditions holds:
   \begin{enumerate}[(1)]
   \item $G'$ is an empty graph;
   \item $G'$ has two components, one being a complete graph, and the other a threshold graph;
   \item $G'$ is connected, and it remains a paired threshold graph after adding a universal vertex.
   \end{enumerate}
\end{lemma}
\begin{proof}
  Let $S$ be the set of all universal vertices deleted in the process.
  Note that $S \neq \emptyset$ because $G$ is connected and $G \neq G'$.

  Suppose first that $G$ is a paired threshold graph.  Note that each component of $G'$ contains at least two vertices: All isolated vertices have been deleted.  If $G'$ has more than two components, then $G$ contains a subgraph isomorphic to Figure~\ref{fig:v+3K2}, which is impossible.  On the other hand, it is straightforward when the number of components in $G'$ is 0 or 1.  Now that $G'$ has two components, denoted by $H_1$ and $H_2$.
  If neither of $H_1$ and $H_2$ is a clique, then we can find from each of them an induced path on three vertices, and thus $G$ contains a subgraph isomorphic to Figure~\ref{fig:v+2P3}.  
  We may assume without loss of generality that $H_1$ is a clique.  If $H_2$ contains a $2K_2$ or $P_4$, then we can find in $G$ a subgraph isomorphic to Figure~\ref{fig:v+3K2} or Figure~\ref{fig:v+P2P4} respectively.  Since $G$ is an interval graph (Theorem~\ref{thm:PTisInterval}), $H_2$ is $\{P_4, C_4, 2K_2\}$-free, hence a threshold graph.

  We now prove the if direction. 

  (1) If $G'$ is empty, then by Proposition~\ref{lem:deletion-sequence}, $G$ is a threshold graph.

  (2) Let $H_1, H_2$ be vertex sets of the two components of $G'$, such that $G[H_1]$ is a complete graph and $G[H_2]$ a threshold graph. 
  By Proposition~\ref{lem:deletion-sequence}, $G - H_1$, which contains $H_2$ and all the deleted vertices, is a threshold graph.  Let $I$ be a maximum independent set of $G - H_1$.  Then $G - I$ is a unit interval graph consisting of two maximal cliques, $V(G) \setminus (H_1\cup I)$ and $H_1 \cup S$. 
It is easy to verify that the neighborhoods of vertices in $I$ form a total order under the containment relation.
   We can construct an umbrella ordering of $G - I$, where $H_2 \setminus I$ come first, followed by $S$, and $H_1$ are in the end.
   Since $H_2 \setminus I$ and $S$ are both true-twin classes of $G-I$, by Proposition~\ref{prop:UIGSubOrderGeneralize}, we can adjust the umbrella ordering such that $N(v)$ appears consecutively for each $v \in I$.  Therefore, $I \uplus (V(G) \setminus I) $ is a paired threshold partition of $G$, and hence $G$ is a paired threshold graph by Theorem~\ref{thm:Characterization}. 
  
  (3) Let $D$ denote the set of all isolated vertices deleted in the process.  We argue that $G - D$ is a paired threshold graph.  In a weight assignment to the graph obtained from $G'$ by adding a universal vertex, which is a paired threshold graph by assumption, the universal vertex has to receive a weight $\ge T_\pm/2$.  Since $S$ are true twins in $G - D$; we can use the same weight to all of them.  
Now that $G - D$ is a paired threshold graph, let $I \uplus U$ be a paired threshold partition of $G - D$, with an umbrella ordering $\sigma$ of $G[U]$ specified in Definition~\ref{def:pt-partition}.  Note that $S \subseteq U$.  As true twins in $G-D$, vertices in $S$ appear consecutively in $\sigma$ by Proposition~\ref{prop:UIGSubOrderGeneralize}.  Moreover, in the graph $G$, the neighborhoods of vertices in $I \cup D$ form a total order under the containment relation.
  Therefore, we can adjust the ordering of $S$ in $\sigma$ such that $N(v)$ appears consecutively for each $v \in D$.
  Since $S = N(D) \subseteq N(x)$ for all $x \in I$, we can verify that parition $(I \cup D) \uplus U$ of $V(G)$ and the ordering $\sigma$ after adjustment satisfy Definition~\ref{def:pt-partition}.  Thus, $G$ is a paired threshold graph by Theorem~\ref{thm:Characterization}.
\end{proof}

We then remove exhaustively universal vertices and isolated vertices from $G$ and study the resulting graph $G'$.  By Lemma~\ref{lemma:ComponentsOfPT}, we return ``yes'' if $G'$ is empty, and return ``no'' if $G'$ contains more than two components.  If $G'$ has precisely two components, then we return whether they satisfy Lemma~\ref{lemma:ComponentsOfPT}(2), i.e.,  one component being a complete graph and the other a threshold graph.  In the rest $G'$ is connected; if $G' \ne G$, then we add a universal vertex to $G'$.  Note that it is an induced subgraph of $G$.
We build a clique path $K_1, \ldots, K_p$ for this subgraph.
Let $I_L$ be greedily obtained as follows.  From $i = 1, \ldots, p$, we pick a simplicial vertex from each $K_i$, as long as $N[I_L]$ still induces a threshold graph.  Let $I_R$ constructed in a similar way, but from $K_p$ to $K_1$.
By Proposition~\ref{prop:CanonicalProperty} and Theorem~\ref{thm:CharacOfCliquePath}, $G$ is a paired threshold graph if and only if one of $I_L$ and $I_R$ satisfies Definition~\ref{def:pt-partition}.  It remains to verify whether one of the partitions is correct.  As usual, $n$ and $m$ denote the number of vertices and the number of edges, respectively.

\begin{figure}[h!]
  \tikz\path (0,0) node[draw=gray!50, text width=.9\textwidth, rectangle, rounded corners, inner xsep=20pt, inner ysep=10pt]{
    \begin{minipage}[t!]{\textwidth} \small
      {\sc Input}: a connected interval graph $G$, and a partition $I \uplus U$ satisfying Definition~\ref{def:pt-partition}(1).
  \\
  {\sc Output}:  whether the partition satisfies Definition~\ref{def:pt-partition}(2--3) or not.
  \begin{tabbing}
    AAA\=AAa\=AAA\=AAA\=MMMMMMMMMMMMMAAAAAAAAAAAAAAAAAAAAAAAAA\=A \kill
    1.\> compute an umbrella ordering $\sigma$ of $G - I$;
    \\
    2.\> let ${\cal T} = T_1, \ldots T_t$ be the true-twin classes of $G - I$ in the order of $\sigma$;
    \\
    3. \> {\bf If} $N(I) \nsubseteq N[T_1]$ and $N(I) \nsubseteq N[T_t]$ {\bf then return} ``no'';
    \\
    4.\> {\bf If} $N(I) \nsubseteq N[T_1]$ {\bf then} reverse the sequence $\cal T$;
    \\
    5.\> {\bf for} $i\leftarrow 1,\ldots, |I|$ {\bf do}
    \\
    5.1.\>\> find the first set $T_\ell$ in $\cal T$ intersecting $N(u_i)$;
    \\
    5.2.\>\> find the last set $T_r$ in $\cal T$ intersecting $N(u_i)$;
    \\
    5.3.\>\> {\bf if} $T\not\subseteq N(u_i)$ for any set $T$ between $T_\ell$ and $T_r$ in $\cal T$ {\bf then return} ``no'';
    \\
    5.4.\>\> replace $T_{\ell}$ by $T_{\ell}\setminus N(u)$ and $T_{\ell} \cap  N(u)$ in order;
    \\
    5.5.\>\> replace $T_{r}$ by $T_{r}\cap N(u)$ and $T_{r} \setminus  N(u)$ in order;
    \\
    6.\> {\bf return} ``yes.'' 
  \end{tabbing}
    \end{minipage}
  };
\caption{The procedure for verifying whether a partition satisfies the conditions in Definition~\ref{def:pt-partition}.}
\label{fig:alg-verify-partition}
\end{figure}

\begin{lemma}\label{lemma:CheckPartition}
  Let  $G$ be a connected interval graph.  Given a partition $I \uplus U$ of $V(G)$ that satisfies Definition~\ref{def:pt-partition}(1), we can check in $O(m + n)$ time whether it satisfies Definition~\ref{def:pt-partition}(2--3) as well.
\end{lemma}
\begin{proof}
  We may number vertices in $I$ in a way that $I = \{ u_1, \ldots, u_{|I|} \}$ and $N(u_{1}) \subseteq \ldots \subseteq N(u_{|I|})$; this is possible because $I$ is an independent set and $N[I]$ induces a threshold graph.   Note that $d(u_1) \le \ldots \le d(u_{|I|})$.  We call the procedure given in Figure~\ref{fig:alg-verify-partition}.

In the first two steps, it starts with finding an umbrella ordering $\sigma$ of $G-I$, and then lists the true-twin classes of $G - I$ in their order of occurrences in $\sigma$.
By Proposition~\ref{prop:UIGSubOrderGeneralize}, the first vertex of any umbrella ordering of $G - I$ has to be from $T_1$ or $T_t$.  If $N(I) \nsubseteq N[T_1]$ and $N(I) \nsubseteq N[T_t]$, then Definition~\ref{def:pt-partition}(3) cannot be satisfied by any umbrella ordering of $G - I$.  This justifies step~3.  Note that it is possible $N(I) \subseteq N[T_t]$ and $N(I) \subseteq N[T_1]$, then $N(I)$ are universal vertices of $G - I$.  This is the trivial case.  Otherwise we make sure that $N(I)\subseteq N(v)$ for some vertex in the first set of $\cal T$ in step~4.  

It now enters step~5.  The sets $T_{\ell}$ and $T_{r}$ exist because $G$ is connected and $I$ is an independent set.  
The focus is on step~5.3.  Suppose that $T\not\subseteq N(u_i)$.  Note that $T$ was not split from the same twin class as $T_\ell$ or $T_r$: Otherwise, $T\subseteq N(u_j)$ for some $j < i$, but then $T\subseteq N(u_i)$ as $N(u_j)\subseteq N(u_i)$.  By Proposition~\ref{prop:UIGSubOrderGeneralize}, vertices in $T$ have to be between vertices of $T_\ell$ and $T_r$ in any umbrella ordering of $G - I$.  Therefore, Definition~\ref{def:pt-partition}(2) cannot be satisfied by any umbrella ordering of $G - I$.  This justifies Step~5.3.  We prove the correctness of step~6 by arguing that If the procedure passes step~5, then it satisfies Definition~\ref{def:pt-partition}(2-3).  We use the umbrella ordering of $G - I$ from $\cal T$ by replacing each set by an arbitrary ordering.  Condition (2) is satisfied for vertex $u_i$ after the $i$th iteration, and the sets containing $N(u_i)$ would never be touched after that.  On the other hand, step~5 never switches the order of two sets, and hence $N(I)\subseteq N[v]$ for any vertex in the first set, which is a subset of $T_1$.  Hence, condition (3) is satisfied as well.

It remains to show that the algorithm can be implemented in $O(n + m)$ time.  It is straightforward for steps 1--4, and hence we focus on step~5.
For each set in $\cal T$, we maintain a doubly linked list; further, we connect these lists into another doubly linked list.  This allows us, among others, to split in time proportional to the number of elements to be split from a set.
We also maintain an array of size $n$, of which the $i$th element points to the position of the $i$th vertex in the doubly linked lists.  With these data structures it is straightforward to implement step~5 in $O(n + m)$ time.  In the first iteration, we go through all the $n$ vertices to find $T_\ell$ and $T_r$; after that we scan from the  $T_\ell$ and $T_r$ of the previous iteration to the left and to the right respectively.  Hence, in $i$th iteration with $1< i\le |I|$, we scan only $O(d(u_i))$ vertices.
This completes the proof.
\end{proof}

We summarize our algorithm in Figure~\ref{fig:alg-recognition}.

\begin{figure}[h!]
  \tikz\path (0,0) node[draw=gray!50, text width=.9\textwidth, rectangle, rounded corners, inner xsep=20pt, inner ysep=10pt]{
    \begin{minipage}[t!]{\textwidth} \small
      \begin{tabbing}
        AAA\=AAa\=AAA\=AAA\=MMMMMMMMMMMMMAAAAAAAAAAAAAAAAAAAAAAAAA\=A \kill
        1.\> {\bf for each} component $C$ of $G$ {\bf do}
        \\
        \>\> {\bf if} $C$ is an unit interval graph {\bf then} remove $C$ from $G$;
        \\
        2.\> {\bf if} $G$ is empty {\bf then return} ``yes'';
        \\
        3.\> {\bf if} $G$ has more than one component {\bf then return} ``no'';
        \\
        4.\> {\bf while} $G$ has a universal or isolated vertex $v$ {\bf then} $G \leftarrow G - v$;
        \\
        5.\> {\bf if} $G$ is empty {\bf then return} ``yes'';
        \\
        6.\> {\bf if} $G$ is not connected {\bf then}
        \\
        \>\> {\bf if} there are more than two components {\bf then return} ``no'';
        \\
        \>\> {\bf if} one component is a clique, and the other is a threshold graph
        \\
        \>\> {\bf then return} ``yes'';
        \\
        \>\> {\bf else return} ``no'';
        \\
        7.\> {\bf if} a vertex is deleted in step~4, {\bf then} add a universal vertex to $G$;
        \\
    8.\> build a clique path $K_1, \ldots, K_p$ of $G$;
    \\
    9.\> build $I_L$ and $I_R$ by greedily taking simplicial vertices from the clique path;
    \\
    10.\> {\bf if} either of $I_L$ and $I_R$ defines a partition satisfying Definition~\ref{def:pt-partition}
    \\
    \> {\bf then return} ``yes'';
    \\
    \> {\bf else return} ``no.''
  \end{tabbing}
    \end{minipage}
  };
  \caption{The algorithm for recognizing paired threshold graphs.}
  \label{fig:alg-recognition}
\end{figure}

\begin{proof}[Proof of Theorem~\ref{thm:RecognizePTG}]
  We use the algorithm described in Figure~\ref{fig:alg-recognition}.  We first prove its correctness.  Steps~1--3 follow from Proposition~\ref{prop:uig-assignment} and Lemma~\ref{lem:disconnected}.  Steps~4--7 follow from Lemma~\ref{lemma:ComponentsOfPT}.  Steps~8--10 follow from Proposition~\ref{prop:CanonicalProperty} and Theorem~\ref{thm:CharacOfCliquePath}.

  We now analyze its running time.  Step~1 calls the $O(m+n)$-time algorithm in \cite{deng-96-proper-interval-and-cag}.  Steps~2 and 3 are trivial.  In step~4, note that isolated vertices and universal vertices can be decided simply by vertex degrees.  Step~5 is trivial.  Step~6 can also be easily done in time $O(m+n)$: It suffices to check the vertex degrees. Step~7 is trivial.  Step~8 calls an algorithm of \cite{hsu-99-recognizing-interval-graphs} or \cite{habib-00-LBFS-and-partition-refinement}, and step~9 is straightforward.  Step~10 calls the $O(m+n)$-time procedure of Lemma~\ref{lemma:CheckPartition}.  This concludes the proof.
\end{proof}

\bibliographystyle{plainurl}

\end{document}